\documentclass[preprint,11pt,hidelinks]{elsarticle}
\usepackage{a4wide}
\usepackage{amsmath}
\usepackage{amsthm}
\usepackage{html}
\usepackage{url}
\usepackage{epsfig}
\usepackage{psfrag}
\usepackage{amssymb}
\usepackage{wrapfig}
\usepackage{algorithm}
\usepackage{algorithmic}
\usepackage{paralist}
\usepackage{xcolor}
\definecolor {processblue}{cmyk}{0.96,0,0,0}

\newcommand{\donotshow}[1]{}

\newcommand{\ignore}[1]{}

\newlength{\setspacing}
\newcommand{\mbegin}{\{\ \ }
\newcommand{\mend}{\}}

\newlength{\mleftindent}
\newlength{\mindent}
\newlength{\mboxwidth}
\newcommand{\mincrement}{\addtolength{\mboxwidth}{-\mindent}}
\newcommand{\mdecrement}{\addtolength{\mboxwidth}{\mindent}}

\newlength{\preprogramskip}
\newlength{\postprogramskip}
\newlength{\mexpwidth}
\newlength{\mexpindent}
\newcommand{\indentafterkeyword}{\hspace*{0.5em}}

\newcommand{\mslifelse}[3]  
{\setlength{\mexpwidth}{\mboxwidth}%
\settowidth{\mexpindent}{{\bf if\indentafterkeyword}}%
\addtolength{\mexpwidth}{-\mexpindent}%
{\bf if\indentafterkeyword}\parbox[t]{\mexpwidth}{#1}\\
\mincrement \mbegin \parbox[t]{\mboxwidth}{#2 \mend} \mdecrement \\
{\bf else} \\
\mincrement \mbegin \parbox[t]{\mboxwidth}{#3}\\
\mend \mdecrement
}

{\vspace{-0.3em}\begin{itemize}
\setlength{\itemsep}{0.2\itemsep}%
\setlength{\parskip}{0.2\parskip}%
\setlength{\topsep}{0.0\topsep}%
}
{\end{itemize}}

{\begin{itemize}
\setlength{\itemsep}{0.2\itemsep}%
\setlength{\parskip}{0.2\parskip}%
\setlength{\topsep}{0.0\topsep}%
}
{\end{itemize}}

{\begin{itemize}
\setlength{\itemsep}{1.5\itemsep}%
\setlength{\parskip}{1.5\parskip}%
\setlength{\topsep}{0.0\topsep}%
}
{\end{itemize}}

{\begin{itemize}
\setlength{\itemsep}{1.5\itemsep}%
\setlength{\parskip}{1.5\parskip}%
\par \smallskip
}
{\end{itemize}}

{\begin{itemize}
\setlength{\itemsep}{1.5\itemsep}%
\setlength{\parskip}{1.5\parskip}%
\par \medskip
}
{\end{itemize}}

{\vspace{-0.3em}\begin{description}
\setlength{\itemsep}{0.2\itemsep}%
\setlength{\parskip}{0.2\parskip}%
\setlength{\topsep}{0.0\topsep}%
}
{\end{description}}

{\begin{description}
\setlength{\itemsep}{0.2\itemsep}%
\setlength{\parskip}{0.2\parskip}%
\setlength{\topsep}{0.0\topsep}%
}
{\end{description}}

{\vspace{-0.3em}\begin{enumerate}
\setlength{\itemsep}{0.2\itemsep}%
\setlength{\parskip}{0.2\parskip}%
\setlength{\topsep}{0.0\topsep}%
}
{\end{enumerate}}

{\begin{enumerate}
\setlength{\itemsep}{0.2\itemsep}%
\setlength{\parskip}{0.2\parskip}%
\setlength{\topsep}{0.0\topsep}%
}
{\end{enumerate}}

\newlength{\proofpostskipamount}\newlength{\proofpreskipamount}
\newlength{\mydefwidth}
\newlength{\mytextwidth}

\newcommand{\myurl}[1]{{\footnotesize \url{#1}}}

\usepackage{mathptmx}
\usepackage[scaled]{helvet}   
\usepackage{courier}

\usepackage[T1]{fontenc}
\usepackage[utf8]{inputenc}

\newcommand{\DM}{Duan-Mehlhorn:Arrow-Debreu-Market}
\newcommand{\DGM}{Duan-Garg-Mehlhorn}
\newcommand{\quotes}[1]{``#1''}

\newtheorem{theorem}{Theorem}[section]

\newtheorem{lemma}[theorem]{Lemma}

\begin{document}

\begin{frontmatter}
\title{Improved Balanced Flow Computation Using Parametric Flow}
\author[MPI]{Omar Darwish\corref{cor1}}
\ead{odarwish@mpi-inf.mpg.de}
\author[MPI]{Kurt Mehlhorn}
\ead{mehlhorn@mpi-inf.mpg.de}
\cortext[cor1]{Corresponding author}
\address[MPI]{Max-Planck-Institute f{\"u}r Informatik, Saarbr{\"u}cken, Germany}

\begin{abstract} We present a new algorithm for computing balanced flows in equality networks arising in market equilibrium computations. The current best time bound for computing balanced flows in such networks requires $O(n)$ maxflow computations, where $n$ is the number of nodes in the network [Devanur et al. 2008]. Our algorithm requires only a single parametric flow computation. The best algorithm for computing parametric flows [Gallo et al. 1989] is only by a logarithmic factor slower than the best algorithms for computing maxflows. Hence, the running time of the algorithms in [Devanur et al. 2008] and [Duan and Mehlhorn 2015] for computing market equilibria in linear Fisher and Arrow-Debreu markets improve by almost a factor of $n$.
\end{abstract}


\end{frontmatter}

\section{Introduction}
Balanced flows are an important ingredient in equilibrium computations for the linear Fisher and Arrow-Debreu markets. In the Arrow-Debreu market model, we have $n$ agents, where each agent owns some goods and a utility function over the existing set of goods. The goal is to assign prices to all the goods such that the market clears, meaning that all the goods are sold and each agent spends all of his budget on only goods with maximum utility. These prices are called the equilibrium prices and the state of the market with such prices is the equilibrium state we want to reach. Under some assumptions, it was shown in \cite{ArrowDebreu} that such prices exist. The Fisher model is a simplification of the Arrow-Debreu model, where each agent comes with a budget instead of the endowment of goods \cite{Fisher:PhD}. The linear Arrow-Debreu market is a special case where the utility functions are linear, and the same applies to the linear Fisher model. Later on, we will refer to agents as buyers. 

The problem of finding the equilibrium prices in these market models has been studied for a long time, see~\cite{DevanurGV13} and~\cite{Codenotti:2004} for surveys. The concept of a balanced flow was introduced in the first combinatorial algorithm~\cite{BalancedFlow} for finding the equilibrium prices in the linear Fisher market. Later, balanced flows were used in \cite{\DM} and~\cite{\DGM} for combinatorial algorithms for computing equilibrium prices in the linear Arrow-Debreu market.

The equilibrium algorithms mentioned above compute balanced flows in a special type of flow network, known as \textit{equality network}. An \textit{equality network} is a bipartite flow network where we have a set of buyers $B$ adjacent to the source node $s$ and a set of goods $C$ adjacent to the sink node $t$. Hence, the vertex set of the network is $\{s,t\} \cup B \cup C$. Also, we assume that the number of buyers is equal to the number of goods. Each buyer $b_i$ has a positive budget $e_i$ and each good $c_j$ has a positive price $p_j$. The edge set is defined as follows:
\begin{enumerate}
\item An edge $(s, b_i)$ with capacity $\mathit{e_i}$ for each $b_i \in B$. 
\item An edge $(c_i, t)$ with capacity $\mathit{p_i}$ for each $c_i \in C$. 
\item All edges running between $B$ and $C$ have infinite capacity. Also, we assume that each buyer is connected to at least one good and similarly each good is connected to at least one buyer.
\end{enumerate}
 A flow in this network corresponds to the money flow from buyers to goods. So, an amount of $x$ units flowing from buyer $b_i$ to good $\mathit{c_j}$ indicates that $b_i$ spends $x$ units on good $\mathit{c_j}$. We refrain from relating the equality network to the market equilibrium computation as this is not important for understanding our algorithm.

Balanced flows are defined for equality networks in~\cite{BalancedFlow}. We denote the set of buyers by $B = \{b_1, b_2, ..., b_n\}$. The surplus $r(b_i)$ of the buyer $b_i$ with respect to a maximum flow $f$ is the residual capacity of the edge $(s,b_i)$ (capacity of this edge minus the flow running through it). The surplus vector is the vector of the surpluses of all buyers $(r(b_1), r(b_2), ..., r(b_n))$. The balanced flow is the maximum flow that minimizes the two-norm of the surplus vector among all maximum flows of the network. In \cite{BalancedFlow}, it is shown that a balanced flow can be computed with $O(n)$ maxflow computations. This is the best time bound achieved up till now for computing balanced flows. 
 
The name \quotes{balanced flow} comes from the fact that it is the maximum flow in the equality network that balances the surpluses as much as possible. Additionally, it is shown in \cite{BalancedFlow} that the surplus vector is unique for all balanced flows in a given equality network. From this fact and other properties shown in \cite{BalancedFlow} and \cite{\DM}, we can deduce the following characterization of balanced flows in equality networks:
\begin{enumerate}
\item Buyers can be partitioned into maximal disjoint blocks based on their surpluses, $B_1, B_2, ..., B_h$ having surpluses $r_1 > r_2 > ... > r_h \geq 0$. Buyers in block $B_i$ have the surplus $r_i$. The uniqueness of this partition follows from the uniqueness of the surplus vector of the buyers.
\item We can assume without loss of generality that $r_h = 0$. 
\item Define $C_i$ to be the set of goods to which money flows from $B_i$. For $i < h$, all goods in $C_i$ are completely sold, i.e., the edges from $C_i$ to $t$ are saturated. There is no flow from $B_i$ to $\mathit{C_j}$ for $j < i$. Also, there are no edges from $B_i$ to $C_j$ for $j > i$.
\end{enumerate}

We show how to compute balanced flows in equality networks by a single parametric flow computation, which improves the running time of the algorithms in~\cite{BalancedFlow} and~\cite{\DM} for computing market equilibria in linear Fisher and Arrow-Debreu markets by almost a factor of $n$. Our algorithm adds to the applications of parametric flows mentioned in~\cite{ParametricFlow}.

In Section~\ref{PF}, we give an introduction to parametric network flows. In Section~\ref{Algo}, we state our algorithm for computing balanced flows, analyse its running time, and prove its correctness.

\section{Parametric Network Flows} \label{PF}
The parametric network flow problem~\cite{ParametricFlow} is a generalization of the standard network flow problem, in which the arc capacities are not fixed, but are functions of a real valued single parameter $\lambda$. More specifically, we consider the following parametric problem:
\begin{enumerate}
\item The arc $(s,b_i)$ has capacity $\max(0,e_i - \lambda)$.
\item The capacity of all other arcs is constant.
\item The parameter $\lambda$ decreases from a large value down to zero. 
\end{enumerate}

We define the \textit{min-cut capacity function} $\kappa(\lambda)$  as the capacity of a minimum cut in the network as a function of the parameter $\lambda$ \cite{ParametricFlow}. Among the cuts of capacity $\kappa(\lambda)$, let $(X(\lambda),\overline{X(\lambda)})$ be the cut with the smallest sink side\footnote{With respect to set inclusion, there is a unique minimum cut of capacity $\kappa(\lambda)$ with smallest sink side.} $\overline{X(\lambda)}$. It is shown in~\cite{ParametricFlow} that $X(\lambda) \subseteq X(\lambda')$ for $\lambda \ge \lambda'$. 

This implies that we have at most $n-1$ distinct minimum cuts. Under the assumption of linearity of the arc capacity functions of $\lambda$, $\kappa(\lambda)$ is a piecewise-linear concave function with at most $n-2$ breakpoints. A breakpoint is a value of $\lambda$ where the slope of $\kappa(\lambda)$ changes. 

As long as the minimum cut is the same, decreasing $\lambda$ will only result in increasing the flow linearly in the current segment of $\kappa(\lambda)$. \cite{ParametricFlow} gives an algorithm for computing all breakpoints of $\kappa(\lambda)$ which runs in $O(nm \log(n^2/m))$ worst-case time complexity, for a network with $n$ nodes and $m$ edges; they call it the \textit{breakpoint algorithm}. Additionally, they state that this algorithm can be easily augmented to store for each vertex $v \notin \{s,t\}$ the breakpoint at which $v$ moves from the sink side to the source side of a minimum cut of minimum sink size without altering the time bound. We will refer to the breakpoint algorithm with the stated augmentation as the \textit{augmented breakpoint algorithm}. We will use this \textit{augmented breakpoint algorithm} in our algorithm for computing balanced flows.

The algorithm in~\cite{ParametricFlow} extends the standard preflow algorithm by Goldberg and Tarjan~\cite{Goldberg-Tarjan} for solving the maximum flow problem. Its running time is only by the factor $\log(n^2/m)$ slower than the running time of the best maxflow algorithms~\cite{Cheriyan-Hagerup,King-Rao-Tarjan:MaxFlow,Orlin:MaxFlow}.

\section{Improved Algorithm for Balanced Flow Computation} \label{Algo}
In this section, we describe our algorithm for computing a balanced flow in a given equality network $N$ and prove its correctness. The intuition behind our algorithm is as follows. For $\lambda = \infty$, all edges out of $s$ have capacity zero. Since all edges into $t$ have positive capacity, the minimum cut will have all buyers on the sink side. Then, step by step, at $\lambda$ values corresponding to $r_1, ..., r_h$, the blocks $B_1, ..., B_h$ of buyers will move from the sink side of the minimum cut to the source side. The behavior of the function $\kappa(\lambda$) changes at these breakpoints, since at these surpluses some buyers will not be able to push more flow to the goods side. Hence the flow will stop at this point at these buyers, while continuing to increase for other buyers, and so on. An example for the evolution of the minimum cut is shown in Figure~\ref{fig:BalancedFlow}. We will next give the details.

\subsection{The Algorithm}
Our algorithm works on an equality network $N$ as follows:
\begin{enumerate}
\item Construct the parametric network $\mathit{PN}$ from $N$. First, construct $\mathit{PN}$ as a copy of $\mathit{N}$. Next, change the capacity of each edge $(s, b_i)$ in $\mathit{PN}$ from $e_i$ to $\max(0, e_i - \lambda)$. 
\item Run the \textit{augmented breakpoint algorithm} on $\mathit{PN}$, in order to find all the breakpoints of $\kappa(\lambda)$ for $\mathit{PN}$ and record for each buyer $b_i$ the breakpoint at which $b_i$ moves from the sink side to the source side of a minimum cut of minimum sink size; let it be $\lambda_i$ for each $b_i$. Note that, initially all buyers are on the sink side of the minimum cut as shown in Figure~\ref{fig:BalancedFlow}. This will be explained in details in the analysis. 
\item Construct the network $N'$ which is a copy of $\mathit{N}$, but with changing the capacity of each edge $(s, b_i)$ from $e_i$ to $e_i - \lambda_i$. Note that, we are setting the capacity of each edge running from $s$ to buyer $b_i$ to be $e_i$ minus the parameter value of the breakpoint where this buyer moves in the minimum cut from the sink side to the source side. This breakpoint of buyer $b_i$ is the surplus of this buyer in the balanced flow of $N$, as we will show in the analysis.
\item Compute the maximum flow in $N'$. It is a balanced flow of the initial network $N$.
\end{enumerate}

\subsection{Complexity and Correctness Analysis}
First, we analyze the worst-case time complexity of our algorithm. We can easily see that steps (1) and (3) run in time complexity equivalent to the size of the network, so steps (2) and (4) are dominating, in terms of time complexity. As explained in the previous section, step (2) runs in $O(nm \log(n^2/m))$ worst-case time complexity, for an equality network with $n$ nodes and $m$ edges. For step (4), it needs one maxflow computation, so it is dominated by step (2). Hence, the overall worst-case time complexity of our algorithm is $O(nm \log(n^2/m))$.\\

Now, we prove the correctness of the algorithm. Recall from the previously stated characterization of balanced flows that we have a unique partition of buyers $B_1, B_2, ..., B_h$ having surpluses $r_1 > r_2 > ... > r_h = 0$, where buyers in block $B_i$ have the surplus $r_i$. Also, recall that we define $C_i$ to be the set of goods to which money flows from $B_i$. Using this characterization, we can deduce the following lemmas.
\begin{lemma}
\label{MinCut}
For $i < h$ and $\lambda \in (r_{i+1}, r_i]$, the minimum cut of minimum sink size corresponding to $\kappa(\lambda)$ in the parametric network $\mathit{PN}$ is $(s \cup (B_1 \cup B_2 \cup ... \cup B_i) \cup (C_1 \cup C_2 \cup ... \cup C_i),$ $t \cup (B_{i+1} \cup B_{i+2} \cup ... \cup B_h) \cup (C_{i+1} \cup C_{i+2} \cup ... \cup C_h))$.
\end{lemma}
\begin{proof}
Let $(X,\overline{X})$ be a minimum cut of minimum sink size of capacity $\kappa(\lambda)$ and let $f$ be a balanced flow in $N$.

We will first prove that the set of buyers $D = B_1 \cup B_2 \cup ... \cup B_i \subseteq X$. Notice that $D$ consists of all buyers whose surplus with respect to $f$ is greater than or equal to $\lambda$. Assume for the sake of a contradiction that $K = D \cap \overline{X}$ is non-empty. Let $L$ be the set of goods in $\overline{X}$ that are connected by an edge to a buyer in $K$; buyers in $K$ may also be connected to goods in $C \cap X$ and $f$ may carry nonzero flow from $K$ to $C \cap X$. From the characterization of balanced flows, $f$ does not carry flow from $B \setminus D$ to $L$. Since $(X,\overline{X})$ is a cut, there can be no edges from buyers in $D \setminus K$ to $L$; note that such edges have infinite capacity. We conclude that in $f$ all flow to the goods in $L$ comes from the buyers in $K$. Since all goods in $L$ are completely sold with respect to $f$ (note that $i < h$), we conclude
\[   \sum_{c_j \in L} p_j \le \sum_{b_j \in K} (e_j - r_j) \le \sum_{b_j \in K} (e_j - \lambda),\]
where the last inequality follows from $\lambda \le r_i$. Thus moving $K \cup L$ from $\overline{X}$ to $X$ will not increase the capacity of the cut and will reduce the sink side of the cut. We have now shown that $D \subseteq X$. Since the edges from $D$ to $\cup_{j \le i} C_j$ have infinite capacity and each good in $\cup_{j \le i} C_j$ is incident to one such edge, we also have $\cup_{j \le i} C_j \subseteq X$. 

We will next prove that $\overline{D}  = B_{i+1} \cup B_{i+2} \cup ... \cup B_h = B \setminus D \subseteq \overline{X}$. Assume for the sake of a contradiction, that $F = X \cap \overline{D}$ is nonempty.  Let $M$ be the goods that are connected to $F$. Since the edges from $F$ to $M$ have infinite capacity, we have $M \subseteq X$, and hence the edges running from $M$ to $t$ contribute to the capacity of the cut. We know that the capacity of these edges is greater than or equal to $\sum_{b_j \in F}{e_j - r_j}$, as in the balanced flow the goods in $M$ receive $\sum_{b_j \in F}{e_j - r_j}$ units of flow from $F$. Since the surplus of each of the buyers in $F$ with respect to the balanced flow is less than $\lambda$, we have $\sum_{b_j \in F}{e_j - \lambda} < \sum_{b_j \in F}{e_j- r_j} \le \sum_{c_j \in M} p_j$. Thus moving $F$ and $M$ from the source to the sink side of the cut will reduce the capacity of the cut (convince yourself that moving $F$ and $M$ to the sink gives a cut). This proves that all buyers in $\overline{D}$ belong to the sink side of the minimum cut. Naturally, all goods adjacent to buyers in $\overline{D}$ without edges from other buyers will also be in the sink side, hence the sink side of the minimum cut contains the nodes $(t \cup (B_{i+1} \cup B_{i+2} \cup ... \cup B_h) \cup (C_{i+1} \cup C_{i+2} \cup ... \cup C_h))$. This concludes the proof.
\end{proof}

\begin{figure}[tb!]
  \centering
  \def\svgwidth{\columnwidth}
  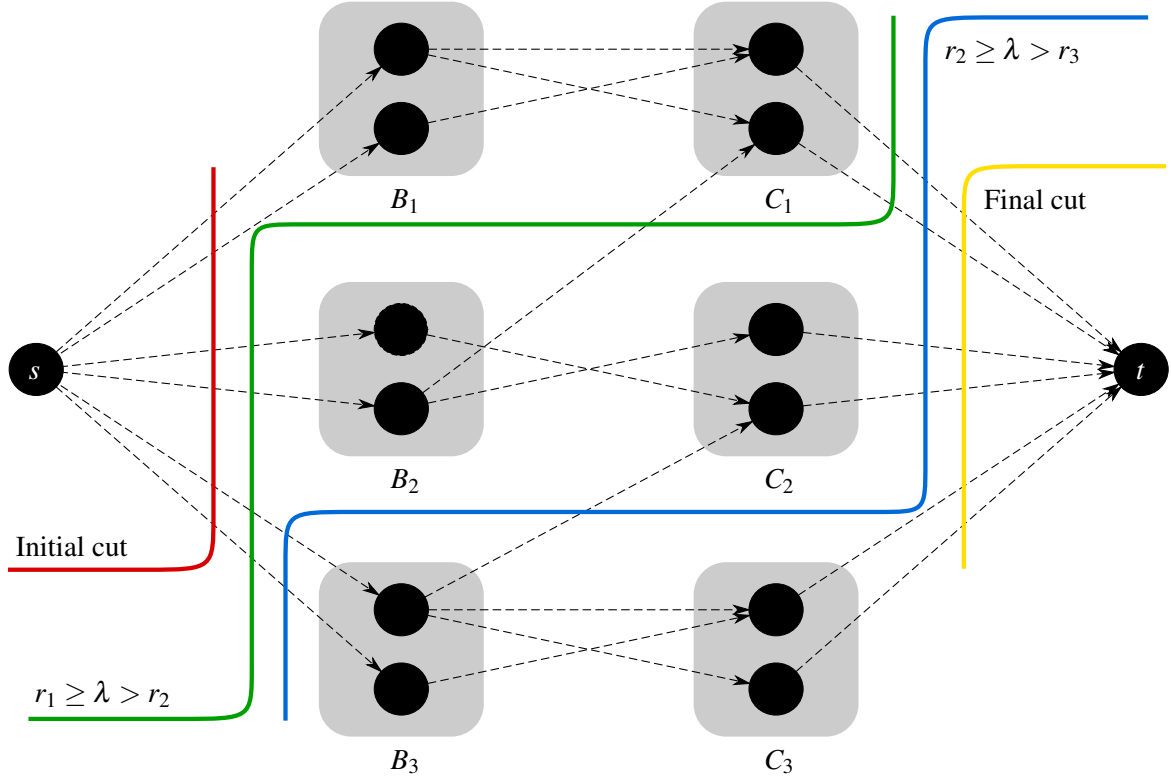
\caption{An example of an equality network in the form of a balanced flow partition. We have 3 blocks $B_1,$ $B_2,$ and $B_3$ of buyers having surpluses $r_1 > r_2 > r_3$, where buyers in block $B_i$ have surplus $r_i$. The figure shows how the minimum cut is changing throughout our algorithm based on the parameter value $\lambda$. Initially at $\lambda = \infty$, the minimum cut includes only $s$ in the source side. As $\lambda$ decreases, the minimum cut will change based on the surpluses. At $\lambda = r_1$, $B_1$ and $C_1$ move from the sink side to the source side, then similarly $B_2$ and $C_2$ move from the sink side to the source side at $\lambda = r_2$. The final minimum cut includes only $t$ in the sink side.}
\label{fig:BalancedFlow}
\end{figure}

\begin{lemma}
\label{breakpoints}
The breakpoints of $\kappa(\lambda)$ for the network $\mathit{PN}$ correspond to the surpluses of the blocks in any balanced flow of $N$.
\end{lemma}
\begin{proof}
Let us decrease $\lambda$ from infinity to zero. We see from lemma \ref{MinCut} that the first breakpoint we will meet is $\lambda = r_1$. At this point all the buyers in $B_1$ will move in the minimum cut of minimum sink size corresponding to $\kappa(\lambda)$ from the sink side to the source side. Similarly, $\lambda = r_2$ will be our second breakpoint,  where all buyers from $B_2$ will move from the sink side to the source side of the minimum cut. The same argument will apply for all the breakpoints.
\end{proof}

\begin{theorem}
The algorithm computes a balanced flow of the equality network $N$, having $n$ nodes on each side and $m$ edges, in $O(nm \log(n^2/m))$ worst-case time complexity.
\end{theorem}
\begin{proof}
We proved the worst-case time complexity in the beginning of this section. Concerning the correctness of the algorithm, we proved in Lemma \ref{breakpoints} that the breakpoints of $\kappa(\lambda)$ for the network $\mathit{PN}$ correspond to the surpluses of the blocks in any balanced flow of $N$. Hence, computing these surpluses will tell us the surplus of each buyer in the balanced flow network and that's what we do in steps (1) and (2) of the algorithm. Now, step (3) enforces these surpluses of each buyer by subtracting this surplus from the capacity of each edge running from the source to the buyer, constructing the network $N'$. So, when we compute the maximum flow in $N'$, all edges from the source node to the buyers will be saturated, hence each buyer will have the corresponding correct surplus. Also, this surplus vector is unique and corresponds to the balanced flow, hence this maximum flow in $N'$ is a balanced flow of $N$.
\end{proof}

\section*{Acknowledgements}
We would like to thank the anonymous reviewer for his/her constructive comments, in particular, with respect to Figure~\ref{fig:BalancedFlow}.

\begin{thebibliography}{11}
\providecommand{\natexlab}[1]{#1}
\providecommand{\url}[1]{\texttt{#1}}
\providecommand{\urlprefix}{URL }
\providecommand{\bibinfo}[2]{#2}

\bibitem[{Devanur et~al.(2008)Devanur, Papadimitriou, Saberi, and
  Vazirani}]{BalancedFlow}
\bibinfo{author}{N.~R. Devanur}, \bibinfo{author}{C.~H. Papadimitriou},
  \bibinfo{author}{A.~Saberi}, \bibinfo{author}{V.~V. Vazirani},
  \bibinfo{title}{Market Equilibrium via a Primal--dual Algorithm for a Convex
  Program}, \bibinfo{journal}{J. ACM}
  \bibinfo{volume}{55}~(\bibinfo{number}{5}) (\bibinfo{year}{2008})
  \bibinfo{pages}{22:1--22:18}. 

\bibitem[{Duan and Mehlhorn(2015)}]{Duan-Mehlhorn:Arrow-Debreu-Market}
\bibinfo{author}{R.~Duan}, \bibinfo{author}{K.~Mehlhorn},
  \bibinfo{title}{\htmladdnormallink{A Combinatorial Polynomial Algorithm for
  the Linear {A}rrow-{D}ebreu Market}{http://arxiv.org/abs/1212.0979}},
  \bibinfo{journal}{Information and Computation} \bibinfo{volume}{243}
  (\bibinfo{year}{2015}) \bibinfo{pages}{112--132}, \bibinfo{note}{a
  preliminary version appeared in ICALP 2013, LNCS 7965, pages 425-436}.

\bibitem[{Gallo et~al.(1989)Gallo, Grigoriadis, and Tarjan}]{ParametricFlow}
\bibinfo{author}{G.~Gallo}, \bibinfo{author}{M.~D. Grigoriadis},
  \bibinfo{author}{R.~E. Tarjan}, \bibinfo{title}{A Fast Parametric Maximum
  Flow Algorithm and Applications}, \bibinfo{journal}{SIAM J. Comput.}
  \bibinfo{volume}{18} (\bibinfo{year}{1989}) \bibinfo{pages}{30–55}.

\bibitem[{Cheriyan and Hagerup(1995)}]{Cheriyan-Hagerup}
\bibinfo{author}{J.~Cheriyan}, \bibinfo{author}{T.~Hagerup}, \bibinfo{title}{A
  Randomized Maximum-Flow Algorithm}, \bibinfo{journal}{SIAM Journal of
  Computing} \bibinfo{volume}{24}~(\bibinfo{number}{2}) (\bibinfo{year}{1995})
  \bibinfo{pages}{203--226}.

\bibitem[{King et~al.(1994)King, Rao, and Tarjan}]{King-Rao-Tarjan:MaxFlow}
\bibinfo{author}{V.~King}, \bibinfo{author}{S.~Rao},
  \bibinfo{author}{R.~Tarjan}, \bibinfo{title}{A Faster Deterministic Maximum
  Flow Algorithm}, \bibinfo{journal}{Journal of Algorithms}
  \bibinfo{volume}{17}~(\bibinfo{number}{3}) (\bibinfo{year}{1994})
  \bibinfo{pages}{447 -- 474}, ISSN \bibinfo{issn}{0196-6774}.

\bibitem[{Orlin(2013)}]{Orlin:MaxFlow}
\bibinfo{author}{J.~B. Orlin}, \bibinfo{title}{Max Flows in O(Nm) Time, or
  Better}, in: \bibinfo{booktitle}{Proceedings of the Forty-fifth Annual ACM
  Symposium on Theory of Computing}, STOC '13, \bibinfo{publisher}{ACM},
  \bibinfo{address}{New York, NY, USA}, ISBN \bibinfo{isbn}{978-1-4503-2029-0},
  \bibinfo{pages}{765--774}, 
  \bibinfo{year}{2013}.

\bibitem[{Arrow and Debreu(1954)}]{ArrowDebreu}
\bibinfo{author}{K.~J. Arrow}, \bibinfo{author}{G.~Debreu},
  \bibinfo{title}{Existence of an Equilibrium for a Competitive Economy},
  \bibinfo{journal}{Econometrica} \bibinfo{volume}{22}~(\bibinfo{number}{3})
  (\bibinfo{year}{1954}) \bibinfo{pages}{265--290}, ISSN
  \bibinfo{issn}{00129682, 14680262}.

\bibitem[{Fisher(1891)}]{Fisher:PhD}
\bibinfo{author}{I.~Fisher}, \bibinfo{title}{Mathematical Investigations in the
  Theory of Value and Prices}, Ph.D. thesis, \bibinfo{school}{Yale University},
  \bibinfo{year}{1891}.

\bibitem[{Devanur et~al.(2013)Devanur, Garg, and V{\'{e}}gh}]{DevanurGV13}
\bibinfo{author}{N.~R. Devanur}, \bibinfo{author}{J.~Garg},
  \bibinfo{author}{L.~A. V{\'{e}}gh}, \bibinfo{title}{A Rational Convex Program
  for Linear {A}rrow-{D}ebreu Markets}, \bibinfo{journal}{CoRR}
  \bibinfo{volume}{abs/1307.8037}.

\bibitem[{Codenotti et~al.(2004)Codenotti, Pemmaraju, and
  Varadarajan}]{Codenotti:2004}
\bibinfo{author}{B.~Codenotti}, \bibinfo{author}{S.~Pemmaraju},
  \bibinfo{author}{K.~Varadarajan}, \bibinfo{title}{The Computation of Market
  Equilibria}, \bibinfo{journal}{SIGACT News}
  \bibinfo{volume}{35}~(\bibinfo{number}{4}) (\bibinfo{year}{2004})
  \bibinfo{pages}{23--37}, ISSN \bibinfo{issn}{0163-5700}.


\bibitem[{Duan et~al.(2015)Duan, Garg, and Mehlhorn}]{Duan-Garg-Mehlhorn}
\bibinfo{author}{R.~Duan}, \bibinfo{author}{J.~Garg},
  \bibinfo{author}{K.~Mehlhorn}, \bibinfo{title}{\htmladdnormallink{An Improved
  Combinatorial Algorithm for the Linear {A}rrow-{D}ebreu
  Market}{http://arxiv.org/abs/1510.02694}},
\newblock In {\em Proceedings of the Twenty-Seventh Annual ACM-SIAM Symposium
  on Discrete Algorithms}, SODA '16, pages 90--106. SIAM, 2016.


\bibitem[{Goldberg and Tarjan(1988)}]{Goldberg-Tarjan}
\bibinfo{author}{A.~Goldberg}, \bibinfo{author}{R.~E. Tarjan},
  \bibinfo{title}{A new Approach to the Maximum-Flow Problem},
  \bibinfo{journal}{Journal of the ACM} \bibinfo{volume}{35}
  (\bibinfo{year}{1988}) \bibinfo{pages}{921--940}.

\end{thebibliography}

\end{document}